\newif\ifdraft

\draftfalse

\ifdraft
  \documentclass[envcountsect,envcountsame,runningheads]{llncs}
  \usepackage[in]{fullpage}
\else
  \documentclass[envcountsect,envcountsame,runningheads]{llncs}
  \usepackage[in]{fullpage}
\fi

\usepackage[T1]{fontenc}
\usepackage[utf8]{inputenc}
\usepackage[dvipsnames]{xcolor}
\definecolor{darkblue}{rgb}{0,0,0.5}
\usepackage{lmodern}
\usepackage{beton}
\usepackage{hyperref}
\hypersetup{
 pdfencoding=auto,         %
 bookmarksnumbered=true,   %
 colorlinks=true,          %
 linkcolor=darkblue,       %
 citecolor=darkblue,       %
 filecolor=darkblue,       %
 urlcolor=darkblue,        %
}
\usepackage{amssymb,amsfonts}
\usepackage{amsmath}
\usepackage{mathtools}

\usepackage{amsthm}
\usepackage{bm}
\usepackage{stmaryrd}
\usepackage{braket}
\usepackage{paralist}
\usepackage[normalem]{ulem}
\usepackage{float}
\usepackage{breakcites}
\usepackage{tcolorbox}
\usepackage{wrapfig}
\usepackage{tikz}
\usetikzlibrary{arrows}
\usetikzlibrary{hobby,backgrounds,calc,trees,decorations.pathmorphing}
\usepackage{enumitem}
\usepackage[
n,
advantage,
operators,
sets,
adversary,
landau,
probability,
notions,
logic,
ff,
mm,
primitives,
events,
complexity,
asymptotics,
keys,
]{cryptocode}
\renewcommand{\secpar}{\lambda}
\usepackage[capitalise]{cleveref}

\makeatletter
\renewcommand\subsubsection{\@startsection{subsubsection}{3}{\z@}%
                       {-18\p@ \@plus -4\p@ \@minus -4\p@}%
                       {0.5em \@plus 0.22em \@minus 0.1em}%
                       {\normalfont\normalsize\bfseries\boldmath}}
\makeatother
\makeatletter
\renewcommand{\@Opargbegintheorem}[4]{%
  #4\trivlist\item[\hskip\labelsep{#3#2\@thmcounterend}]}
\makeatother

\Crefname{experiment}{Experiment}{Experiments}
\Crefname{construction}{Construction}{Constructions}
\Crefname{conjecture}{Conjecture}{Conjectures}
\let\llncssubparagraph\subparagraph
\let\subparagraph\paragraph
\RequirePackage{titlesec}
\let\subparagraph\llncssubparagraph
\titleformat*{\paragraph}{\normalsize\bfseries}

\ifdraft

\fi

\newcommand{\Uppi}{\mathrm{\Pi}}
\newcommand{\Upphi}{\mathrm{\Phi}}
\newcommand{\Uppsi}{\mathrm{\Psi}}

\newcommand{\qpt}{\mathsf{QPT}}

\newcommand{\ooracle}{\mathcal{O}}
\newcommand{\goodO}{\mathcal{O}} %
\newcommand{\X}{\mathcal{X}} %
\newcommand{\Y}{\mathcal{Y}} %
\newcommand{\hset}{\mathcal{H}} %
\newcommand{\hsup}{\widehat{\supp}^H(\ket{\phi})} %
\newcommand{\hmax}[2][H]{\hat{h}^{#1}_{max}\left(#2\right)}

\renewcommand{\epsilon}{\varepsilon}

\newcommand{\goodstate}[1][\yspace, \delta, d, N]{(#1)}

\newcommand{\eve}{\mathsf{Eve}}
\newcommand{\alice}{\mathsf{A}}
\newcommand{\bob}{\mathsf{B}}

\newcommand{\jstate}{\Uppsi} \newcommand{\simulatedA}{\phi_A^E} %
\newcommand{\simulatedB}{\phi_B^E} %
\newcommand{\realA}{\phi_{\alice}} \newcommand{\realB}{\phi_{\bob}}
\newcommand{\m}{\psi} %
\newcommand{\mi}[1][i]{\psi_{#1}} %
\newcommand{\qi}[1][i]{q_{#1}} %
\newcommand{\simulatedm}{\psi^E} %
\newcommand{\h}{h} %
\newcommand{\Aop}{\ensuremath{\alice_{\mathsf{fin}}}} %
\newcommand{\Bop}{\ensuremath{\bob_{\mathsf{fin}}}}
\newcommand{\simulatedoutput}{{\tau}}
\newcommand{\simulatedoutputu}{\widetilde{\tau}} %
\newcommand{\realoutput}{\Phi_A}
\newcommand{\keeps}{\varepsilon} %
\newcommand{\Pik}[1][\key]{\Uppi_{#1}}
\newcommand{\keerr}{\nu} %
\newcommand{\kA}{\ensuremath{\key_{\alice}}}
\newcommand{\kB}{\ensuremath{\key_{\bob}}}
\newcommand{\khA}{\ensuremath{\key_{\alice'}}}
\newcommand{\evestate}{\ensuremath{\ket{\mathrm{\Psi}_{t}^{\eve}}}}
\newcommand{\jointstate}{\ensuremath{\ket{\mathrm{\widehat{\Psi}}_{t, L}}}}
\newcommand{\jjointstate}[1][]{\ensuremath{\ket{\mathrm{\widehat{\Psi}^{(i)}}_{t, L#1}}}}

\newcommand{\ketbra}[2]{{\vert#1\rangle\!\langle#2\vert}}
\newcommand{\kb}[1]{\ketbra{#1}{#1}}
\newcommand{\bk}[2]{\braket{#1|#2}}
\DeclareMathOperator{\Tr}{Tr}

\newcommand{\reg}{}
\newcommand{\unitary}{U}
\newcommand{\ident}{\mathcal{I}}

\let\supp\relax
\DeclareMathOperator{\supp}{supp}
\newcommand{\size}[1]{\left|#1\right|}

\DeclareMathOperator*{\opEE}{\mathbb{E}}
\newcommand{\EE}[1]{\ensuremath{\opEE\left[#1\right]}}

\newcommand{\randvar}[1]{\bm{\mathrm{#1}}}
\renewcommand\prob[2][]{\Pr_{#1}\left[#2\right]}
\DeclareMathOperator{\pr}{Pr} %

\newcommand\given[1][]{\:#1\vert\:}

\newcommand{\xspace}{\mathcal{X}}
\newcommand{\yspace}{\mathcal{Y}}

\newcommand{\eps}{\varepsilon}

\newcommand{\epshl}{\eps} %
\newcommand{\querylist}[1]{\ensuremath{\mathcal{Q}_{#1}}}
\newcommand{\qalist}[1]{#1}

\DeclareMathAlphabet{\mathpzc}{OT1}{pzc}{m}{it}
\newcommand{\qgen}{\ensuremath{\mathpzc{Gen}}}
\newcommand{\qenc}{\ensuremath{\mathpzc{Enc}}}
\newcommand{\qdec}{\ensuremath{\mathpzc{Dec}}}
\newcommand{\qc}{\ensuremath{\mathpzc{qc}}}

\title{Towards the Impossibility of Quantum Public Key Encryption with Classical Keys from One-Way Functions%
}
\ifdraft
  \author{}
  \institute{}
  \date{}
\else
  \author{Samuel Bouaziz--Ermann\inst{1}
    \and Alex B. Grilo\inst{1}
    \and Damien Vergnaud\inst{1}
    \and Quoc-Huy Vu\inst{2}}
    \institute{Sorbonne Universit\'e, CNRS, LIP6, France \and Léonard de Vinci Pôle Universitaire, Research Center, 92 916 Paris La Défense, France}
    \authorrunning{S. Bouaziz--Ermann et al.}
\fi

\begin{document}
\maketitle
\begin{abstract}
  There has been a recent interest in proposing quantum protocols whose security
  relies on weaker computational assumptions than their classical counterparts.
  Importantly to our work, it has been recently shown that public-key encryption
  (PKE) from one-way functions (OWF) is possible if we consider quantum public
  keys.
  Notice that we do not expect classical PKE from OWF given the impossibility
  results of Impagliazzo and Rudich (STOC'89).

  However, the distribution of quantum public keys is a challenging task.
  Therefore, the main question that motivates our work is if quantum PKE from
  OWF is possible if we have classical public keys.
  Such protocols are impossible if ciphertexts are also classical, given the
  impossibility result of Austrin {\em et al.}(CRYPTO'22) of quantum enhanced
  key-agreement (KA) with classical communication.

  In this paper, we focus on black-box separation for PKE with classical public
  key and quantum ciphertext from OWF under the polynomial compatibility
  conjecture, first introduced in Austrin {\em et al.}.
  More precisely, we show the separation when the decryption algorithm of the
  PKE does not query the OWF.
  We prove our result by extending the techniques of Austrin {\em et al.}
  and we show an attack for KA in an extended classical communication model
  where the last message in the protocol can be a quantum state.
\end{abstract}

\section{Introduction}
After decades of focusing on the possibility of information-theoretically secure
quantum protocols, initiated by the land-marking results on money
schemes~\cite{Wiesner83} and key-agreement~\cite{BB84}, there has been recent
progress in understanding how quantum resources can be used to implement
cryptographic primitives under weaker computational assumptions.

More concretely, it has been shown in~\cite{EC:GLSV21,C:BCKM21a} that Oblivious
Transfer (OT) and Multi-party computation (MPC), two central primitives in
cryptography can be constructed from one-way functions (OWF), the weakest
classical cryptographic assumption.
Such a result has been extended to show OT and MPC can be constructed from
pseudo-random quantum states~\cite{C:AnaQiaYue22}, which is expected to be a
weaker computational assumption than OWF~\cite{Kre21,STOC:KQST23}.
This is in stark contrast with the classical case since we do not expect OT and
MPC to be built from one-way functions~\cite{STOC:ImpRud89}.

More recently, it has been asked if quantum protocols are possible for
public-key encryption from OWF (or weaker assumptions).
While the conditional impossibility result for key-agreement
of~\cite{C:ACCFLM22} implies that public-key encryption (PKE) from OWF with classical communication is
impossible even if the honest parties are quantum,\footnote{Such a result is actually conditioned on a conjecture that we state in \Cref{conjecture:pcc} and discuss in
  \Cref{sec:discussion}.} it has been recently shown that PKE
can be constructed from OWF if we have a quantum
public-key~\cite{Coladangelo23,BGHMSVW23,KMNY23}.
However, having a quantum public key is not ideal, given the issues that appear
with public-key distribution, authentication, and reusability.
These results leave then as an open question if quantum PKE from OWF is possible
with a classical public key and quantum ciphertext.

In this work, we extend the result of \cite{C:ACCFLM22} and we show that
key agreement is impossible when Alice and Bob exchange classical messages and
at the very last round, Bob sends a quantum message to Alice.
Our result holds under the same conjecture as~\cite{C:ACCFLM22} but is limited
to the setting where Alice does not query the random oracle in the last round of the protocol.
More concretely, we achieve the following result.

\begin{theorem}[Informal]\label{theorem:ka}
  Let $\Pi$ be a key agreement protocol between Alice and Bob, where they first
  exchange classical messages and at the last round Bob sends a quantum message,
  and Alice and Bob agree on a key $k$.
  Let $n$ be the number of queries that Alice and Bob make to a random oracle
  $\ooracle$.
  Then, assuming Alice does not query the oracle after receiving the quantum
  message from Bob, Eve can find $k$ with $\bigO{\poly[n]}$ classical queries to
  $\ooracle$ with probability $\frac{1}{\poly[n]}$.
\end{theorem}

 With this result in hand, we show that quantum PKE is impossible with a
classical public key in the Quantum Random Oracle Model (QROM), when the decryption algorithm does not query
the random oracle.

\begin{corollary}[Informal]\label{cor:qpke}
  Assume \(\left(\qgen,\qenc,\qdec\right)\) is a Public-Key Encryption scheme, where the
  public key is classical and the ciphertext is a quantum state.
  Assuming the algorithms $\qgen$ and $\qenc$ makes at most $n$ quantum
  queries to a random oracle $\ooracle$, then there exists an algorithm $\eve$
  that can decipher by making $\bigO{\poly[n]}$ classical queries to $\ooracle$.
\end{corollary}
Using known techniques from black-box separation, our results can easily
be translated to give separations of qPKE from black-box OWF.
We also note that our result (\cref{cor:qpke}) marks an initial step towards
proving the conjecture of~\cite{EPRINT:MorYam22c} on the possibility of
black-box constructions of qPKE with classical public keys from quantum
symmetric key encryption.

\subsection{Technical Overview}

To prove \Cref{theorem:ka}, we start with a key-agreement protocol with perfect correctness where Alice and Bob have quantum access to a random oracle and exchange polynomially many rounds of classical messages, and then Bob sends a final quantum message $\ket{\psi}$ to Alice.
\footnote{In this overview, we consider for simplicity that the last message sent by Bob is a pure state,
denoted as \(\ket{\m}\). In our formal proof, we consider the general case in which the message is a mixed state.}

We show an attack where with inverse polynomial probability:
\begin{enumerate}
\item Given the classical transcript and $\ket{\psi}$, Eve guesses the key $k$ that Alice and Bob would share.
\item Eve sends a quantum state $\psi^E$ to Alice such that Alice agrees on the key $k$ at the end of the protocol.
\end{enumerate}

While the first item is sufficient to break the key agreement protocol, the second item allows us to show a much stronger attack: Eve is an
  active adversary that not only retrieves the key but also does it in a way that Alice and Bob will not detect later since both of them share the same key.

To prove the first item, we use the same technique as \cite[Construction
4.10]{C:ACCFLM22}, which queries all of the ``$\epshl$-heavy queries'' to the random oracle.
This approach is similar to the classical approach of \cite{C:BarMah09}, whose
construction queries all of the values that are queried by Alice and Bob with
probability \textit{at least} $\epshl$.
These are called the ``$\epshl$-heavy queries'',
and with the
right parameter of $\epshl$, one can show that it allows Eve to find all of the
\textit{intersection queries} with high probability, that is the values queried by both Alice and Bob.
A problem that appears in the quantum setting is that the notion of
\textit{intersection queries} is unclear, as Alice and Bob are allowed to query
the oracle in superposition, it is thus hard to precisely define what
information Alice and Bob know about the oracle.
This means that a definition for quantum intersection queries
  must be such that the knowledge of the intersection queries is sufficient to find the key with high probability, which is a strong property.
In the quantum attack of \cite{C:ACCFLM22}, they start by defining the
\textit{quantum} heavy queries. Roughly, these are the queries with high amplitude (see
\Cref{def:heavyquery} for a formal definition), which is the natural quantum equivalent of the
classical definition that appears in \cite{C:BarMah09}.
To replace the notion of intersection queries, which is problematic in the quantum setting, they propose the Polynomial Compatibility Conjecture (PCC). This conjecture implies that if a pair of quantum states satisfy some conditions, then there exists a random oracle that is consistent with the transcript and with both quantum states.
In their attack, Eve learns all of the heavy queries and generates a simulation of the states of Alice and Bob.
They show that if $\eve$ is not able to retrieve the key with high probability from these simulated states, then the protocol does not have perfect correctness, leading to a contradiction. This is shown by proving that if Eve does not find the correct key, by the PCC, there exists an oracle $h$ that is consistent with both states,  and therefore there exists an execution of the protocol where Alice outputs keys $1$ and Bob outputs key $0$.

We extend this result to the case where the last message is quantum.
Similar to \cite[Construction 4.10]{C:ACCFLM22}, we define the quantum-heavy
query learner algorithm, which is formally defined in
\Cref{construction:simulation} and whose goal is to query all of the
$\epshl$-heavy queries.
In this overview, we consider for simplicity that the last message sent by Bob is a pure state,
denoted as \(\ket{\m}\).
\(\eve\) then run Alice's last step of the protocol (which is publicly defined by the protocol)
on the simulated internal state of Alice that Eve generated and the
quantum message from Bob.
We then need to show that Eve will still be
able to guess the good key with high probability.
This translates to showing that, for some noticeable parameter $\keerr$:
\begin{equation}
  \label{eq:findthekey}
  \Tr\left(\Pik \Aop \left(\kb{\simulatedA}_{W'_A} \otimes  \kb{\h}_H \otimes \kb{\m}_M \right) \left(\Aop\right)^{\dagger}\right) \geq 1 - \keerr,
\end{equation}
where the register $W'_{A}$ contains $\eve$'s simulated state of Alice
$\ket{\simulatedA}$, register $H$ contain the superposition of all possible oracles
that are consistent with $\eve$'s knowledge of the real oracle, and register $M$ contains the
message $\ket{\m}$ that Bob sent to Alice.
The unitary $\Aop$ corresponds to Alice's operation in the last step of the
protocol, after she receives Bob's message, and the projector $\Pik$ measures
the key.
\Cref{eq:findthekey} translates to saying that given the \textit{real} message
that Bob sends to Alice, Eve can find Bob's key by applying the
operations that Alice would have applied on the simulated state of Alice
$\ket{\simulatedA}$ that she obtained by using the quantum-heavy queries
learner.
This equation is proven in~\cref{sec:findthekey}.

The proof of \Cref{eq:findthekey} comes from the fact that since Alice does not
query the oracle when she applies the operator \(\Aop\) after receiving the
quantum message \(\m\) from Bob, then the register \(H\) is unchanged and thus
the resulting state keeps the properties necessary to apply the PCC.

For the second item, we need to define the state $\simulatedm$ that Eve sends to Alice. The idea is the following: Eve will pick the post-measurement state of the measurement described in \Cref{eq:findthekey}, and she applies $\Aop^\dagger$ to it. Then, Eve traces out the registers $W'_A$ and $H$ and $\simulatedm$ is the remaining state in register $M$.

To show that Alice computes the
  same key as Bob and Eve with high probability, we show that $\simulatedm$ is close to $\ket{\psi}$:
  \begin{equation}
  \label{eq:closemessage}
  \bra{\m}\simulatedm \ket{\m} \geq 1 - \keerr.
\end{equation}

Using \Cref{eq:closemessage} and the perfect correctness of the protocol, we
then show that Alice and Bob will agree on the same key with high probability.
This corresponds to proving the following inequality:
\begin{equation}
  \label{eq:agrees}
  \Tr\left(\Pik \Aop (\kb{\realA}_{W_{A}} \otimes \kb{\h}_H \otimes \simulatedm) \left(\Aop\right)^{\dagger}\right) \geq 1 - \keerr,
\end{equation}
where $\ket{\realA}$ is Alice's real internal state and $\simulatedm$ is the
message that Eve sends to Alice.
\Cref{eq:agrees} translates to saying that given the message of $\eve$, Alice
will find the same key as $\eve$ with high probability when she does her final computation.
These equations are proven in~\cref{section:breakingccqc}.

Finally, \Cref{cor:qpke} follows from the fact that if public key encryption with quantum ciphertexts is
possible, then we can construct a key agreement protocol: Alice sends the public key and Bob answers with the encryption of a random key $k$.

\subsection{Related Works, Discussion and Open Problems}
\label{sec:discussion}

\paragraph{The Polynomial Compatibility Conjecture.} First introduced in
\cite{C:ACCFLM22}, the Polynomial Compatibility Conjecture (PCC) is already
known to imply separation results for key agreement~\cite{C:ACCFLM22} and
non-interactive commitments~\cite{EC:ChuLinMah23}.
The conjecture has an alternative expression that uses polynomials and is
equivalent to the statement in \Cref{conjecture:pcc}.
The PCC is known to be true with exponential parameters~\cite{C:ACCFLM22}, but it is still open
with polynomial parameters.
Proving it would be interesting as it would now also establish the separation result
for quantum PKE, along with potentially more results as it is a
strong statement.

\paragraph{Quantum Public Key Encryption.} 
Classically, public key encryption
(PKE) cannot be constructed from black-box one-way functions~\cite{STOC:ImpRud89}.
In the quantum context, various definitions of quantum PKE exist, leading to
different feasibility outcomes.
With quantum public keys and classical ciphertexts, quantum PKE can be
constructed from one-way functions~\cite{Coladangelo23,BGHMSVW23,KMNY23}.
However, it remains unclear how
the distribution of such public keys could be effectively distributed in practice among different parties.
Our result focuses on quantum PKE with classical public
keys and quantum ciphertexts.
In this setting, the distribution of public keys could be implemented using currently available public key infrastructure (PKI).
Moreover, compared to having a quantum public key and a classical ciphertext, having a classical public key and a quantum ciphertext is less problematic for implementations, as the message is
supposed to be received by only one party and thus the potential destruction of
the message after the decryption is inconsequential.
With this definition of quantum PKE, we achieve a step towards proving a similar result as the classical case.

\paragraph{Classical Communication One Quantum Message Key Agreement Protocols.} In our
paper, we introduce a scheme that we call Classical Communication One Quantum
Message Key Agreement (CC1QM-KA) protocols.
In these types of protocol, Alice and Bob communicate classically, except for
the last message that is quantum.
We show that key agreement is impossible with this type of protocol in the QROM
if Alice does not query the random oracle after receiving the last message.
One natural question is what happens if we allow the \textit{first} message to
be quantum, while the rest of the communication is classical.
Interestingly enough, \cite{BB84} falls into this category of protocol, thus key
agreement is possible \textit{unconditionally} in this setting.
This asymmetry in terms of feasibility results is quite surprising and a
possible explanation is the fact that we cannot postselect on quantum messages, i.e. generate a state that is consistent with an algorithm and one of its output register being a specific quantum state.
Indeed, the classical $\eve$ attacks imply a simulation of the internal state of
the parties that is consistent with the message, which corresponds to computing an internal state postselected on the classical messages that are communicated. %
With a quantum message, this would be possible with a classical description of the quantum message since $\eve$ is unbounded,
but it is non-trivial with only the quantum state as $\eve$ must learn what the quantum message is in the first place somehow. %
However, in the CC1QM setting, we do not need to do this postselection, as a
simulation of the last part of the protocol is enough to find the right key.

\paragraph{Allowing oracle queries in the decryption algorithm}
To prove the stronger result that qPKE is impossible even when the decryption
algorithms query the oracle, one needs to show an attack on CC1QM-KA protocols
where Alice queries the oracle in the last part of the protocol.
At first glance, one may think that \Cref{eq:findthekey} should be true even if Alice makes queries to
the oracle in \(\Aop\), because every new information about the oracle that she
learns at this stage of the protocol will not be transmitted to Bob since there
is no communication afterward.
However, some issues that do not appear in \cite{C:ACCFLM22} arise when trying
to prove such an inequality.

The first (natural) problem is that since the last message is quantum, Eve cannot compute the heavy queries (which would be sufficient for the attack). Therefore, we need to find another way of simulating Alice's last oracle calls without learning the heavy queries.

A first attempt is to use the operator \(\Aop^{\goodO}\), that corresponds to
Alice's computation in the last step of the protocol with the \emph{real} oracle
\(\goodO\).
Because this corresponds to the operation that the real Alice would have done and the real outcome is deterministic (since the protocol has perfect correctness), it could allow Eve to find the real key. However, the problem in this approach is that Eve has her simulated state that was constructed using a {\em simulated} oracle (with correct values for heavy-queries) and Alice's algorithm could use some consistency check that would fail when we decide to change the oracle.

On the other hand, if we want to use the simulated oracle instead of the real
oracle, then there is a trivial protocol for which the attack does not work.
In this protocol, Bob just picks a random
value \(x\in\X\), queries it, and sends \(\ket{x}\) to Alice. Alice and Bob agree then on the key  \(H(x)\).
By using the simulation oracle, Eve would not be able to find the key with non-negligible probability.

While these two complications are artificial since they do not lead to a secure protocol, they put a barrier to finding a common attack that would make Eve find the keys from Alice and Bob.

\section{Preliminaries}
\subsection{Notation}
The following notations will be used throughout the paper,
\begin{itemize}
  \item By \(\secpar\) we denote the security parameter.
  \item We use calligraphic letters (e.g., \(\mathcal{X}\)) to denote sets.
    We use \(\yspace^{\xspace}\) to denote the set of all functions from
    \(\xspace\) to \(\yspace\).
  \item We use bold letters (e.g., \(\randvar{m}\)) to denote random variables
    and distributions.
    We write \(m \sample \randvar{m}\) to denote that \(m\) is sampled from the
    distribution \(\randvar{m}\).
    We write \(m \sample \mathcal{M}\) to denote that \(m\) is sampled uniformly from the set \(\mathcal{M}\).
  \item We use the Dirac notation for pure states, e.g., \(\ket{\psi}\), while
    mixed states will be denoted by lowercase Greek letters, e.g., \(\rho\).
\end{itemize}
For the basics of quantum computation, we refer readers
to~\cite{nielsen2010quantum}.

\subsection{Quantum Computation}

\begin{definition}[Oracle-aided quantum algorithms]
  A quantum algorithm \(\adv\) is a family of quantum circuits
  \(\adv \coloneqq \{A_{\secpar}\}_{\secpar \in \NN}\) that act on three sets of
  registers: input registers \(\reg{X}\), output registers \(\reg{Y}\), and work
  registers \(\reg{Z}\).
  For convenience, we let \(\reg{W} \coloneqq (\reg{X}, \reg{Y}, \reg{Z})\)
  denote the \emph{internal registers} of \(\adv\).
  For each input \(x \in \bin^{\secpar}\), the output is computed by running the
  algorithm \(A_{\secpar}\) on
  \(\ket{x}_{\reg{X}}\ket{0}_{\reg{Y}}\ket{0}_{\reg{W}}\) and at the end the
  output registers are measured in the computational basis to obtain the output.

  A \(d\)-query quantum oracle algorithm \(\adv^{h}\) that has access to an
  oracle \(h\), defined by the unitary \(\ooracle_{h}\) can be specified by a
  sequence of unitary matrices \((\unitary_{d}, \unitary_{d-1}, \ldots, \unitary_{0})\).
  The final state of the algorithm is defined as
  \(\unitary_{d}\ooracle_{h}\unitary_{d-1}\ooracle_{h}\ldots\ooracle_{h}\unitary_{0}\ket{x}_{\reg{X}}\ket{0}_{\reg{Y}}\ket{0}_{\reg{Z}}\).
  When the oracle \(h\) implements some classical function
  \(h: \xspace \to \yspace\), the corresponding query operator \(\ooracle_{h}\) is
  defined as
  \(\ket{x}_{\reg{X}}\ket{y}_{\reg{Y}} \mapsto \ket{x}_{\reg{X}} \ket{y \oplus h(x)}_{\reg{Y}}\).

  When \(\adv^{h}\) is clear from the context, we omit the superscript \(h\) and
  write \(\adv\).
\end{definition}

The following preliminary is borrowed from the formalization
of~\cite{C:ACCFLM22}.

\begin{definition}[The computational and the Fourier basis]
  Let \(\yspace\) be a finite Abelian group with cardinality \(\size{\yspace}\).
  Let \(\{\ket{y}\}_{y \in \yspace}\) be an orthonormal basis of
  \(\CC^{\size{\yspace}}\).
  We refer to this basis as the computational basis.
  Let \(\hat{\yspace}\) be the dual group which is known to be isomorphic to
  \(\yspace\).
  Recall that a member \(\hat{y} \in \hat{\yspace}\) is a character function
  (i.e., a function from \(\yspace\) to the multiplicative group of non-zero
  complex numbers).
  The Fourier basis \(\{\ket{\hat{y}}\}_{\hat{y} \in \hat{\yspace}}\) of
  \(\CC^{\size{\yspace}}\) is defined as
  \begin{equation*}
    \ket{\hat{y}} = \frac{1}{\sqrt{\size{\yspace}}} \sum_{y \in \yspace} \hat{y}(y)^{\dagger}\ket{y}
    \text{ and }
    \ket{{y}} = \frac{1}{\sqrt{\size{\yspace}}} \sum_{\hat{y} \in \hat{\yspace}} \hat{y}(y)\ket{\hat{y}}.
  \end{equation*}
\end{definition}

\begin{definition}[Functions and their (quantum) representations]
  For any function \(h \in \yspace^{\xspace}\), we define its quantum
  representation to be
  \(\ket{h}_{H} \coloneqq \bigotimes_{x \in \xspace}\ket{h(x)}_{H_{x}}\) in the
  computational basis, where the register \(H_{x}\) is associated with
  \(\CC^{\yspace}\) for all \(x \in \xspace\), and the register \(H\) is
  compounded of all \(H_{x}\).
  Similarly, for any \(\hat{h} \in \hat{\yspace}^{\xspace}\) we define
  \(\ket{\hat{h}}_{H} \coloneqq \bigotimes_{x \in \xspace} \ket{\hat{h}(x)}_{H_{x}}\)
  in the Fourier basis.
\end{definition}

Zhandry~\cite{C:Zhandry19} shows that the purified random oracle is perfectly
indistinguishable from the (standard) quantum random oracle, and thus instead of
considering the query operator \(\ooracle_{h}\), we can consider another
equivalent query oracle \(\ooracle\) acting on three registers
\(\reg{X}, \reg{Y}, \reg{H}\) as follows.
\begin{equation*}
  \ket{x}_{\reg{X}}\ket{y}_{\reg{Y}}\ket{h}_{\reg{H}} \mapsto \ket{x}_{\reg{X}}\ket{y \oplus h(x)}_{\reg{Y}}\ket{h}_{\reg{H}},
\end{equation*}
where the oracle register \(H\) is initialized as
\(\ket{\Upphi_{0}}_{H} = \sum_{h \in \hset} \frac{1}{\sqrt{\size{\hset}}}\ket{h}_{H}\).

Note that in the Fourier basis, the unitary \(\ooracle\) acts as follows:
\begin{equation*}
  \ket{x}_{X}\ket{\hat{y}}_{Y}\ket{\hat{h}}_{H} \mapsto \ket{x}_{X}\ket{\hat{y}}_{Y} \bigotimes_{x' \in \X} \ket{\hat{h}(x') - \delta_{x,x'} \cdot \hat{y}}_{H},
\end{equation*}
where \(\delta_{x,x'}\) is equal to \(1\) if \(x=x'\), and \(0\) otherwise, and the
oracle register \(H\) is initialized as
\(\ket{\Upphi_{0}}_{H} = \bigotimes_{x \in \X} \ket{\hat{0}}_{H_{x}}\).

\begin{definition}[Purified view of two-party protocols in the QROM]
  \label{def:purifiedview}
  A two-party protocol in the Quantum-Computation Classical-Computation (QCCC)
  model is a protocol in which two quantum algorithms, Alice and Bob, can query
  the oracle, apply quantum operation on their internal registers, and send
  classical strings over the public (authenticated) channel to the other party.
  The sequence of the strings sent during the protocol is called the
  \emph{transcript} of the protocol.
  Let \(\reg{W}_{A}\) and \(\reg{W}_{B}\) be Alice's and Bob's internal
  registers, respectively.
  Let \(\mathcal{H} \coloneqq \yspace^{\xspace}\).
  For any two-party protocol, we define its purified version as follows.
  \begin{itemize}
    \item If the protocol is inputless, start with
      \(\ket{0}_{\reg{W}_{A}}\ket{0}_{\reg{W}_{B}}\sum_{h \in \mathcal{H}}\frac{1}{\sqrt{\size{\mathcal{H}}}}\ket{h}_{\reg{H}}\).
      Otherwise, if Alice takes as input a classical string \(a \in \xspace\) and
      Bob takes as input a classical string \(b \in \xspace\), start with
      \(\ket{a}_{\reg{W}_{A}}\ket{b}_{\reg{W}_{B}}\sum_{h \in \mathcal{H}}\frac{1}{\sqrt{\size{\mathcal{H}}}}\ket{h}_{\reg{H}}\).
    \item Alice and Bob run the protocol in superposition, that is, all the
      measurements (including those used for generating the transcript) are
      delayed and the query operator \(\ooracle_{h}\) is replaced by
      \(\ooracle\).
    \item Let \(\ket{\jstate}_{\reg{W}_{A}\reg{W}_{B}\reg{H}}\) denote the state
      at the end of the protocol, and let
      \(\ket{\jstate_{t}}_{\reg{W}_{A}\reg{W}_{B}\reg{H}}\) denote the
      post-measurement state of
      \(\ket{\jstate}_{\reg{W}_{A}\reg{W}_{B}\reg{H}}\) which is consistent with
      the transcript \(t\).
  \end{itemize}
\end{definition}

We now define some properties related to this new register for the database $\ket{h}_H$.

\begin{definition}[Non-zero queries in Fourier basis]
  Let \(\yspace\) be a finite Abelian group and \(\hat{\yspace}\) be the dual
  group.
  For any \(\hat{y} \in \hat{\yspace}^{\xspace}\), we define the size of
  \(\hat{h}\) to be
  \begin{equation*}
    \size{\hat{h}} \coloneqq \size{\{ x: x \in \xspace, \hat{h}(x) \neq \hat{0}\}}.
  \end{equation*}
\end{definition}

\begin{definition}[Oracle support]
  \label{def:support}
  Let \(\hat{\mathcal{H}} \coloneqq \hat{\yspace}^{\xspace}\).
  For any vector
  \(\ket{\phi}_{\reg{W}\reg{H}} = \sum_{w, \hat{h} \in \hat{\mathcal{H}}}\alpha_{w, \hat{h}} \ket{w}_{\reg{W}}\ket{\hat{h}}_{\reg{H}}\),
  we define the \emph{oracle support in the Fourier basis} of \(\ket{\phi}\) as
  \begin{equation*}
    \hsup \coloneqq \left\{\hat{h}: \exists w \text{ s.t. } \alpha_{w,\hat{h}} \neq 0 \right\}.
  \end{equation*}
  We denote \(\hmax{\ket{\phi}}\) the function \(\hat{h} \in \hsup\) that has the
  largest size \(\size{\hat{h}}\) (if such function is not unique, by default we
  pick the lexicographically first one).
  The definition extends naturally when the register \(\reg{W}\) does not exist.

  Similarly, if we write the oracle part in the computational basis
  \(\ket{\phi}_{\reg{W}\reg{H}} = \sum_{w, h \in \mathcal{H}} \beta_{w, h}\ket{w}_{\reg{W}}\ket{h}_{\reg{H}}\),
  then we define the oracle support in the computational basis of \(\ket{\phi}\) as
  \begin{equation*}
    \supp^{H}(\ket{\phi}) \coloneqq \left\{h: \exists w \text{ s.t. } \beta_{w, h} \neq 0 \right\}.
  \end{equation*}
\end{definition}

\begin{definition}
  A partial oracle \(L\) is a partial function from \(\X\) to \(\Y\).
  The domain of \(L\) is denoted by \(Q_{L} \coloneqq \textrm{dom}(L)\).
  Equivalently, we view \(L\) as a finite set of pairs \((x, y_{x}) \in \X \times \Y\)
  such that for all \((x, y_{x}), (x', y'_{x}) \in L, x \neq x'\).
  We say a partial oracle \(L\) is consistent with \(h: \X \to \Y\) if and only if
  \(h(x) = y_{x}\) holds for all \(x \in Q_{L}\).

  For any partial oracle \(L\), we define the associated projector \(\Uppi_{L}\)
  by
  \begin{equation*}
    \Uppi_{L} \coloneqq \bigotimes_{x \in Q_{L}} \ketbra{y_{x}}{y_{x}}_{H_{x}} \bigotimes_{x \notin Q_{L}} \ident_{H_{x}},
  \end{equation*}
  where \(\ident_{H_{x}}\) is the identity operator acting on \(H_{x}\).
  It holds that \(\Uppi_{L}\ket{h}_{H} = \ket{h}_{H}\) if \(h\) is consistent
  with \(L\), and \(\Uppi_{L}\ket{h}_{H} = 0\) otherwise.
\end{definition}

\begin{lemma}
  \label{lemma:sparse}
  If \(\alice\) asks at most \(d\) queries to the superposition oracle, then for
  all possible outcomes of \(\alice\)'s intermediate measurements, the joint
  state \(\ket{\psi}_{WH}\) conditioned on the outcome satisfies
  \(\size{\hmax{\ket{\psi}}} \leq d\).
\end{lemma}

\begin{lemma}
  \label{lemma:sparse2}
  Given a state \(\ket{\psi}_{WH}\) and a partial oracle \(L\), the state
  \(\Uppi_{L}\ket{\psi}_{WH}\) can be written as
  \begin{equation*}
    \Uppi_{L}\ket{\psi}_{WH} \coloneqq \sum_{w \in \mathcal{W}, \hat{h} \in \hat{\mathcal{H}}'} \alpha'_{w, \hat{h}} \ket{w}_{W} \bigotimes_{x \notin Q_{L}} \ket{\hat{h}(x)}_{H_{x}} \bigotimes_{x \in Q_{L}}\ket{y_{x}}_{H_{x}},
  \end{equation*}
  where \(\hat{\mathcal{H}}'\) is the set of functions from \(\X \setminus Q_{L}\) to
  \(\hat{\Y}\).
  Furthermore, if \(\size{\hmax{\ket{\psi}}} \leq d\), then
  \(\size{\hmax[H']{\Uppi_{L}\ket{\psi}}} \leq d\), where \(H'\) is the set of
  registers corresponding to \(\X \setminus Q_{L}\).
\end{lemma}

\subsection{Quantum-Heavy Queries Learner}
We now define the \textit{quantum-heavy queries learner} algorithm.
It was first defined in \cite[Construction 4.10]{C:ACCFLM22}, which can be seen
as the quantum counterpart of the classical \emph{independence learner} of
\cite{C:BarMah09}, where Eve learns all the \textit{\(\epshl\)-heavy
  queries} of both Alice and Bob.

\begin{definition}[Quantum \(\epshl\)-heavy queries~{\cite[Definition
    4.9]{C:ACCFLM22}}]
  \label{def:heavyquery}
  For \(x \in \xspace\), define the projector
  \begin{equation*}
    \Pik[x] \coloneqq \sum_{\hat{y} \in \hat{\yspace} \setminus \{\hat{0}\}}\ketbra{\hat{y}}{\hat{y}}_{\reg{H}_x}.
  \end{equation*}
  Given a quantum state \(\ket{\psi}_{\reg{W}_{A}\reg{W}_{B}\reg{H}}\), the
  weight of any \(x \in \xspace\) is defined as
  \begin{equation*}
    w(x) \coloneqq \norm{\Pik[x]\ket{\psi}}^2,
  \end{equation*}
  that is, the quantum heaviness of \(x\) is the probability of obtaining a
  non-\(\hat{0}\) outcome while measuring \(\reg{H}_{x}\) in the Fourier basis.
  We call \(x \in \X\) a \emph{quantum \(\epshl\)-heavy} query if
  \(w(x) \geq \epshl\).
\end{definition}

\begin{construction}[Quantum-heavy queries learner~\cite{C:ACCFLM22}]
  \label{construction:simulation}
  Let \((\alice, \bob)\) be an inputless two-party QCCC protocol relative to a
  random oracle \(h\), in which Alice and Bob make at most \(d\) quantum queries
  to the oracle.
  Given the transcript \(t\), (computationally-unbounded) attacking algorithm \(\eve\) is parameterized by
  \(\epshl\) and works as follows.
  \begin{enumerate}
    \item Let \(L\) denote the set of oracle query-answer pairs obtained by
      \(\eve\) from the oracle, and \(\querylist{L}\) is defined similarly while
      only containing the queries.
      Initially prepare \(\qalist{L} = \emptyset\) and the classical description of the
      state
      \begin{equation*}
        \ket{\psi}_{W'_{\alice}W'_{\bob}H'} = \ket{0}_{W'_{\alice}}\ket{0}_{W'_{\bob}}\ket{\Upphi_0}_{H'},
      \end{equation*}
      where \(\ket{\Upphi_{0}}\) is a uniform superposition over all \(h \in \mathcal{H}\),
      \(W'_{\alice}\) , \(W'_{\bob}\) and \(H'\) are the simulated registers for
      Alice, Bob, and the oracle prepared by \(\eve\).
    \item Simulate the state evolution during the protocol.
      Concretely, \(\eve\) calculates the state in \(W'_{\alice}W'_{\bob}H'\)
      after each round in the protocol.
      Whenever \(\eve\) encounters the moments in which Alice (Bob) sends their
      message, \(\eve\) calculates the post-measurement state that is consistent
      with \(t\).
    \item While there is any query \(x \notin \querylist{L}\) that is quantum
      \(\epshl\)-heavy conditioned on \((t, \qalist{L})\), do the following:
      \begin{enumerate}
        \item Ask the lexicographically first quantum \(\epshl\)-heavy query
          \(x\) from the real oracle \(h\).
        \item Update the state in \(W'_{\alice}W'_{\bob}H'\) to the
          post-measurement state that is consistent with \((x, h(x))\).
        \item Update \(\qalist{L}\) by adding \((x, h(x))\) to \(\qalist{L}\).
      \end{enumerate}
    \item When there is no quantum \(\epshl\)-heavy query left to ask, \(\eve\)
      outputs the simulated quantum state
      \(\ket{\psi_{t}}_{W'_{\alice}W'_{\bob}H'}\) and her list \(\qalist{L}\),
      conditioned on the transcript \(t\).
  \end{enumerate}
\end{construction}

\begin{remark}
  We note that~\cref{construction:simulation} described above is almost
  identical to \cite[Construction 4.10]{C:ACCFLM22}.
  The only difference is that \(\eve\) outputs the simulated state,
  which can be constructed from the classical description that \(\eve\) has
  computed, along with the list of queries she made to the oracle.
\end{remark}

The technical properties of the quantum-heavy queries learner
in~\cref{construction:simulation} are stated in the following lemma.
\begin{lemma}[{\cite{C:ACCFLM22}}]
  \label{lemma:eve}
  For any \(0 < \eps < 1 \), the quantum-heavy queries learner described
  in~\cref{construction:simulation} satisfies the following properties:
  \begin{itemize}
    \item \textbf{Efficiency:} \(\EE{\size{\qalist{L}}} \leq \frac{d}{\eps}\),
      where the expectation is over the randomness of the oracle and the
      algorithm \(\eve\).
    \item \textbf{Security:} When the learner stops and learns a list
      \(\qalist{L}\), there is no \(x \in \querylist{L}\) that is \(\eps\)-quantum
      heavy in the purified view of \(\eve\) conditioned on knowing
      \(\qalist{L}\) and the transcript \(t\).
  \end{itemize}
\end{lemma}

\subsection{Polynomial Compatibility Conjecture}
In this section, we recall the Polynomial Compatibility Conjecture (PCC)
of~\cite{C:ACCFLM22}.
The formulation we use here is based on quantum states.

\begin{definition}[\(\goodstate\)-state~{\cite[Definition 4.1]{C:ACCFLM22}}]
  \label{def:goodstates}
  Let \(H\) be a register over the Hilbert space \({\Y}^{N}\).
  A quantum state \(\ket{\psi}\) over registers \(\reg{W}\) and \(\reg{H}\) is a
  \emph{\(\goodstate\)-state} if it satisfies the following two conditions:
  \begin{itemize}
    \item \textbf{\(d\)-sparsity}: \(\size{\hmax{\ket{\psi}}} \leq d\).
      This means that for any measurement of registers \(H\) in the Fourier
      basis, and \(W\) in any basis, the oracle support in the Fourier basis is
      at most \(d\).
    \item \textbf{\(\delta\)-lightness}: For every \(x \in \xspace\), if we measure the
      \(\reg{H}_x\) register of \(\ket{\psi}\) in the Fourier basis, the
      probability of getting \(\hat{0}\) is at least \(1 - \delta\).
      This mean that \(\ket{\psi}\) has no \(\delta\)-heavy queries.
  \end{itemize}
\end{definition}

\begin{definition}[Compatible states {\cite[Definition 4.2]{C:ACCFLM22}}]
  \label{def:compatibility}
  Two quantum states \(\ket{\phi}\) and \(\ket{\psi}\) over registers \(\reg{W}\) and
  \(\reg{H}\) are \emph{compatible} if their oracle supports in the
  computational basis (as defined in~\cref{def:support}) have non-empty
  intersection, i.e., if \(\supp^H(\ket{\phi}) \cap \supp^H(\ket{\psi}) \neq \emptyset\).
\end{definition}

We now state the conjecture.
\begin{conjecture}[Polynomial compatibility conjecture~{\cite[Conjecture
    4.3]{C:ACCFLM22}}]
  \label{conjecture:pcc}
  There exists a finite Abelian group \(\yspace\) and \(\delta = 1/\poly[d]\) such
  that for any \(d, N \in \NN\), it holds that any two
  \(\goodstate[\yspace, \delta(d), d, N]\)-states \(\ket{\phi}\) and \(\ket{\psi}\) are
  compatible.
\end{conjecture}

\subsection{Useful Lemmas}
We will use the following lemma frequently in our proofs in subsequent sections.
\begin{lemma}[Independence~{\cite[Lemma 3.2]{C:ACCFLM22}}]
  \label{lemma:independency}
  Suppose two quantum algorithms \(\alice\) and \(\bob\) interact classically in
  the quantum random oracle model.
  Let \(\reg{W}_{\alice}\) and \(\reg{W}_{\bob}\) denote their internal
  registers respectively.
  Then, at any time during the protocol, conditioned on the (classical)
  transcript \(t\) and the fixed oracle \(h \in \hset\), the joint state of the
  registers \(\reg{W}_{\alice}\) and \(\reg{W}_{\bob}\) conditioned on \(t\) and
  \(h\) is a product state.
\end{lemma}

\section{Attack on the Key Agreement Protocols}
\label{sect:ka}
In this section, we consider key agreement protocols in an extended setting
where both parties are quantum algorithms but they can only send classical
strings over the public authenticated channel to the other party, except that
the last message in the protocol can be a quantum state (in this case, the last
message is not authenticated).
We call this the Classical Communication One Quantum Message (CC1QM) model.
In this extended setting, we show a conditional result based on the polynomial
compatibility conjecture, that any protocol in the CC1QM model with perfect
completeness where Alice does not query the oracle after receiving the last
message can be broken with an expected polynomial number of queries.
We present the formal definition of key agreement protocols in the CC1QM model
in~\cref{sect:CC1QM-KA}.
In~\cref{sect:attack-ka}, we state the main result and its proof.

\subsection{Definitions}
\label{sect:CC1QM-KA}
We start by defining the model of Classical Communication One Quantum Message,
where two quantum parties (Alice and Bob) communicate using the public
authenticated classical channel, except for the last message that can be
quantum.
We assume the first message is from Alice to Bob, while the last message is from
Bob to Alice, and the last quantum message is non-authenticated.
This can be assumed without loss of generality since if the first message is
from Bob to Alice, we can always transform it into the other case, by letting
Alice sends a dummy message to Bob for the first message.
Furthermore, we consider the case where the key that Alice and Bob agree on is
one bit and the protocol succeeds with probability \(1\) (i.e., perfect
correctness).
Also, as for Quantum Key Distribution (QKD), we allow the parties to abort the
protocol at any time, if they detect suspicious activity in the quantum
communication.
Formally, this is done by making Alice output the character \(\bot\) instead of
a key when the protocol is aborted.
More formally, we define:
\begin{definition}[Key agreement protocols in the CC1QM model]
  \label{def:CC1QM-KA}
  We say that \((\alice,\bob)\) is a key agreement protocol between two parties
  Alice and Bob in the CC1QM model (CC1QM-KA) if the following holds:
  \begin{enumerate}
    \item At the beginning of the protocol, Alice and Bob share no common
      information.
      Their corresponding algorithms, \(\alice\) and \(\bob\), are stateful
      oracle-aided quantum algorithms which make at most \(d\) oracle queries.
    \item \textbf{CC1QM.}
      All of the messages are classical messages, except for the last message
      (from Bob to Alice) that can be a (mixed) quantum state, denoted as
      \(\m\).
      The transcript of the protocol is denoted as \(T \coloneqq \left(m_1,\cdots,m_{\ell},\m\right)\).
    \item \textbf{Perfect completeness.}
      At the end of the protocol, Alice and Bob agree on a key
      \(\key \in \{0,1\}\) with probability \(1\) when the protocol succeeds
      (i.e.~when neither Alice or Bob outputs \(k=\bot\)).
    \item \textbf{Security.}
      Let \(\Aop'\) be Alice's last computation in the protocol after she
      receives the final message from Bob.
      By deferred measurement principle, we can modify \(\Aop'\) so that it
      applies a unitary transformation \(\Aop\) followed by a measurement in the
      computational basis \(\{\Pik\}_{\key \in \bin}\) and outputting a key
      \(\key\), and we write \(\Aop' \coloneqq \Pik \Aop\).
      Similarly, let \(\Bop' \coloneqq \Pik \Bop\) be Bob's last computation in
      the protocol after he sends the final message to Alice.
      We note that \(\Aop\) and \(\Bop\) can make quantum queries to the oracle,
      and the output of \(\Aop\) (resp. \(\Bop'\)) is the output key of Alice
      (resp. Bob) at the end of the
      protocol execution.
      Let \((T, \realA, \realB) \gets \left\langle \alice \models \bob \right\rangle\) be the output of an
      execution of the protocol right before Alice receives the last quantum
      message from Bob, where \(T \coloneqq \left(m_1,\cdots,m_{\ell},\m\right)\) is the transcript of
      the execution, \(\realA\) and \(\realB\) are the internal state of
      \(\alice\) and \(\bob\), respectively.
      \((\alice, \bob)\) is secure if for any polynomially-bounded query
      adversary \(\edv\):
      \begin{align*}
        \pr\left[
        \begin{array}{c}
          \key = \key_{\alice} = \key_{\bob} \\
          \key_{\alice} \neq \bot \\
          \key_{\bob} \neq \bot
        \end{array}
        \ \middle\vert
        \begin{array}{r}
          (T, \realA, \realB) \gets  \left\langle \alice \models \bob \right\rangle \\
          (\key, \m') \gets \edv(1^{\secpar}, T) \\
          \key_{\alice} \gets \Aop'(\realA, \m') \\
          \key_{\bob} \gets \Bop'(\realB)
        \end{array}
        \right]
        \leq \negl.
      \end{align*}
      We say that a CC1QM-KA protocol \((\alice, \bob)\) is \((\epsilon,s)\)-broken if
      there exists an attacker Eve that finds the key of \((\alice, \bob)\) with
      probability at least \(\epsilon\), \((\alice, \bob)\) succeeds with probability
      at least \(\poly[\epsilon]\), and Eve makes an expected number of queries at most
      \(s\).
  \end{enumerate}
\end{definition}

\subsection{The Attack on Key Agreements Protocols}
\label{sect:attack-ka}
The goal of the section is to prove the following theorem that states that are
no CC1QM-KA protocol in the QROM.
\begin{theorem}
  \label{theorem:mainka}
  Let \((\alice,\bob)\) be a CC1QM-KA protocol, where Alice and Bob make at most
  \(d\) queries to a random oracle \(h: \mathcal{X} \rightarrow \mathcal{Y}\),
  and Alice does not query the oracle in the last part of the protocol (after
  receiving the quantum message from Bob).
  Assuming \Cref{conjecture:pcc} is true, then there exists an attacker Eve that
  makes at most \(\poly[d, |\mathcal{Y}|]\) many \emph{classical} queries to \(h\) and
  breaks the security (according to~\cref{def:CC1QM-KA}) with probability at
  least \(0.8\).
\end{theorem}
The proof consists of two parts, the first one shows that Eve manages to find the same key as the one computed as Bob, and this is proven in \Cref{sec:findthekey}.
The second part consists of showing that Alice agrees on the same key as Eve and Bob and this corresponds to \Cref{section:breakingccqc}.
First, in the next section, we prove that the attack does not depend on the group of the domain of the function.

\subsubsection{Group Equivalence of the Attack}
We first show that if there is an attack for an Abelian group \(\mathcal{Y}\), then there
is an attack for any other Abelian group \(\mathcal{Y}'\), up to some error terms.
This allows us to relax the conjecture to be true for \textit{any} Abelian
group, as in \cite{C:ACCFLM22}.
The proof follows closely \cite{C:ACCFLM22}'s proof as they are almost
identical, and we include it here for completeness.
\begin{lemma}%
  \label{lemma:groupequivalence}
  Suppose there exists a finite Abelian group \(\mathcal{Y}\), a constant \(\tau > 0\) and a
  function \(s(\cdot)\) such that for all \(d \in \mathbb{N}\) and any CC1QM-KA protocol \(\left(\alice_1^{h},\bob_1^{h}\right)\) where Alice and Bob asks
  \(d\) queries to a random oracle \(h\) whose range is \(\mathcal{Y}\), and
  Alice does not query the oracle after receiving the last message, it holds that
  \(\left(\alice_1^{h},\bob_1^{h}\right)\) is \((\tau,s(d))\)-broken.
  Then, for any finite Abelian group \(\mathcal{Y}'\), any \(d' \in \mathbb{N}\), \(\delta > 0\) and any
  CC1QM-KA protocol \(\left(\alice_2^{h'},\bob_2^{h'}\right)\)
  where Alice and Bob asks \(d'\) queries to another random oracle \(h'\) whose
  range is \(\mathcal{Y}'\), \(\left(\alice_2^{h'},\bob_2^{h'}\right)\), and
  Alice does not query the oracle after receiving the last message, can be
  \((\tau - \delta, 4s(md'))\)-broken, where
  \begin{equation*}
    m = \ceil{\log_{\size{\mathcal{Y}}}(d'^{3}\size{\mathcal{Y'}}/4\delta^2)}.
  \end{equation*}
\end{lemma}

\begin{proof}
  The proof follows from the proof of Lemma 4.8 from \cite{C:ACCFLM22}.
  The only difference is that we must also show that with probability at
  least \(\tau - \delta\), Alice and Bob agree on the same key as Eve.
  However, their proof relies on the fact that we can simulate a random oracle
  with another random oracle, even when their ranges are different, up to some
  errors.
  Thus, their proof follows through in our setting as well, and with the same
  parameters.
\end{proof}

\begin{lemma}[Attacking CC1QM-KA protocols]
  \label{lemma:breakingccqc}
  Assume \Cref{conjecture:pcc} is true for some Abelian group \(\mathcal{Y}\) and
  parameters \(d\) and \(\delta = \keerr/\keeps\).
  Let \((\alice, \bob)\) be a CC1QM-KA protocol where Alice and Bob make at most
  \(d\) queries to a random oracle \(h:\mathcal{X} \rightarrow \mathcal{Y}\),
  and Alice does not query the oracle after receiving the last message.
  Then, there exists an active attacker Eve who finds the secret key \(k\) with
  probability \(1-\keerr\) according to~\cref{def:CC1QM-KA}.
  Moreover, Eve is expected to make at most \(d/\keeps\) queries to \(h\).
\end{lemma}

The proof of \Cref{lemma:breakingccqc} is given in
subsequent~\Cref{section:breakingccqc,sec:findthekey}.

We can now prove \Cref{theorem:mainka}:
\begin{proof}[Proof of \Cref{theorem:mainka}]
  The proof follows immediately from \Cref{lemma:breakingccqc},
  \Cref{lemma:groupequivalence} and the proof of \cite[Theorem 4.5]{C:ACCFLM22}.
\end{proof}

\subsubsection{Part 1: Finding Bob's Key}
\label{sec:findthekey}
In this subsection, we show that the attack algorithm described
in~\Cref{construction:simulation} can efficiently find Bob's key with high
probability, assuming that \Cref{conjecture:pcc} is true.
We first state and show a useful lemma that allows us to assume that when Bob
sends the last message, he has already computed the key \(k\) on his side.

\begin{lemma}
  \label{lemma:bobmeasures}
  Let \((\alice,\bob)\) be a CC1QM-KA protocol.
  Let \(\realB \) be the internal state of \(\bob\) after he computed the
  message \(\m\).
  Then, we can assume w.l.o.g.~that Bob has computed the key \(\kB\) from
  \(\realB\) before he sends the last message \(\m\) to Alice.
\end{lemma}

\begin{proof}

  Since the last message of the protocol is sent to Alice by
  Bob, by the no-signaling principle, Alice's computation after receiving \(\m\)
  must commute with Bob's computation after sending \(\m\).
  Thus, Bob can compute the key on his side before sending the last message
  \(\m\).

\end{proof}

\begin{lemma}[Simulation]
  \label{lemma:simulation}
  Let \((\alice,\bob)\) be a CC1QM-KA protocol where Alice and Bob make at most
  \(d\) queries to an oracle \(h:\mathcal{X} \rightarrow \mathcal{Y}\), and
  Alice does not query the oracle after receiving the last message.
  Assuming~\cref{conjecture:pcc} is true, for \(0 < \keerr < 1\), there exists
  an active attacker Eve that finds Bob's key \(\kB\) with probability at least
  \(1 - \keerr\) and Eve is expected to make at
  most \(\poly[d, \frac{1}{\keerr}]\) queries to \(h\).
\end{lemma}

\begin{proof}[Proof of~\cref{lemma:simulation}]
  Let Bob's last message be \(\m_{M} = \sum_i\qi\kb{\mi}_{M}\), and let
  \(\Aop' \coloneqq \Pik \Aop\) be Alice's computation in the last step of the
  protocol.
  Let \(\kB\) be the key computed by Bob at the end of the protocol.
  By~\cref{lemma:bobmeasures}, we can assume that Bob already computes his key
  \(\kB\) before sending the last message to Alice.

  Our attacking algorithm \(\eve_{1}\) is described below.
  \begin{construction}
    \(\eve_{1}\) runs the quantum-heavy queries learner \(\eve\)
    in~\cref{construction:simulation} with parameter
    \(\keeps \coloneqq \frac{1}{\poly[d,\frac{1}{\keerr}]}\) 
    conditioned on the classical transcript \(t\) until before Bob sends his
    last message, except that it aborts if \(\eve\) asks more than
    \(\frac{d}{\keeps}\) queries.
    In the case \(\eve_{1}\) does not abort, let
    \(\evestate_{W'_{\alice}W'_{\bob}H'}\) be the state that \(\eve\) outputs,
    conditioned on the classical transcript \(t\).
    \(\eve_{1}\) then outputs the measurement outcome of
    \begin{equation*}
      \Aop' \left(\evestate_{W'_{\alice}W'_{\bob}H'} \otimes \m_{M}\right),
    \end{equation*}
    where \(\Aop'\) makes no oracle query to \(h\) and acts on two registers
    \(W'_{\alice}\) and \(M\) only.
  \end{construction}

  By~\cref{lemma:eve}, the number of queries asked by \(\eve\) satisfies
  \(\EE{\size{\qalist{L}}} \leq \frac{d}{\keeps}\).
  By Markov's inequality, we have
  \begin{equation*}
    \prob{\size{\qalist{L}} \geq \frac{d}{\keerr \cdot \keeps}} \leq \keerr.
  \end{equation*}
  Thus, we can conclude that with probability at least \(1 - \keerr\), all of
  the following events hold:
  \begin{itemize}
    \item \(\eve_{1}\) is efficient: \(\eve_{1}\) does not abort and asks at
      most \(\frac{d}{\keerr \cdot \keeps} = \poly[d, 1 / \keerr]\) queries.
    \item Up until before Bob sends his last message, no quantum
      \(\keeps\)-heavy query is left: for all \(x \notin \querylist{L}, w(x) < \keeps\),
      where \(w(\cdot)\) is defined in~\cref{def:heavyquery}.
  \end{itemize}
  Suppose that all the above events occur for the rest of the
  proof \hypertarget{misc:assumption}{\((^{\mathbf{\star}})\)}.
  For simplicity, denote \(\evestate_{W'_{\alice}W'_{\bob}H'}\) as
  \(\evestate_{W'_{\alice}E}\).

  We will consider the purified version of the protocol.
  Let \(\ket{\phi_{t}}_{WH}\) be the joint state of the real protocol before Bob
  sends his last message to Alice, conditioned on the classical transcript
  \(t\).
  After \(\eve_{1}\) learns the heavy queries, the resulting state becomes
  \(\ket{\phi_{t, L}}\) conditioned on \(t\) and Eve's list of query-answer \(L\).
  Since the oracle registers corresponding to \(Q_{L}\) are now measured, we can
  consider the ``truncated'' version of \(\ket{\phi_{t, L}}_{WH}\) by discarding
  those registers.
  Let \(\widetilde{H} \coloneqq \{H_{x}\}_{x \in \xspace \setminus Q_{L}}\) be the set of
  remaining registers.
  By \(\ket{\phi_{t, L}}_{W\widetilde{H}}\) we denote the truncated
  \(\ket{\phi_{t, L}}_{WH}\).

  Let
  \(\jointstate_{W'_{\alice}EW\widetilde{H}} \coloneqq \evestate_{W'_{\alice}E}\ket{\phi_{t, L}}_{W\widetilde{H}}\)
  be the joint state of \(\eve_{1}\) and the real protocol right before Bob
  sends his last message to Alice.
  By~\cref{lemma:sparse}, it holds that
  \(\size{\hmax{\jointstate}} \leq \poly[d, 1 / \keerr]\), and
  by~\cref{lemma:sparse2}, it holds that
  \(\size{\hmax[\widetilde{H}]{\jointstate}} \leq \poly[d, 1 / \keerr]\).

  By the assumption \hyperlink{misc:assumption}{\((^{\mathbf{\star}})\)} above, we
  have that \(\jointstate\) is a
  \(\goodstate[\Y,\epshl,\poly[d, \keerr],\size{\X}]\)-state (with the register
  \(H\) in~\cref{def:goodstates} being \(\widehat{H}\)).
  Next, let \(\m_{M} = \sum_i\qi\kb{\mi}_{M}\), we need to show that
  \begin{equation*}
    \forall i, \norm{\Pik \Aop \jointstate_{W'_{\alice}EW\widetilde{H}} \ket{\mi}_M}^2 \geq 1 - \frac{1}{\keerr},
  \end{equation*}
  where \(\Aop\) cannot make queries to \(h\) and only acts on \(W'_{\alice}\) and \(M\).

  Fix \(i\) and write
  \(\jjointstate_{W'_{\alice}EW\widetilde{H}M} \coloneqq \Aop \left(\jointstate_{W'_{\alice}EW\widetilde{H}} \otimes \ket{\mi}_{M}\right)\).

  \begin{claim}
    If \(\jointstate\) is a
    \(\goodstate[\Y,\epshl,\poly[d, \keerr],\size{\X}]\)-state, it follows that
    \(\jjointstate\) is a
    \(\goodstate[\Y,\epshl,\poly[d, \keerr],\size{\X}]\)-state as well.
  \end{claim}
  \begin{proof}
    Assume that \(\jointstate_{RH}\) is a
    \(\goodstate[\Y,\epshl,\poly[d, \keerr],\size{\X}]\)-state.
    Then, \(\jointstate_{RH} \otimes \ket{\mi}_{M}\) is also a
    \(\goodstate[\Y,\epshl,\poly[d, \keerr],\size{\X}]\)-state, because this
    property only depends on the $H$ register, who is unchanged there.
    Then, since \(\Aop\) makes no query to the random oracle, the oracle
    register $H$ is not modified and thus \(\jjointstate\) is a \(\goodstate[\Y,\epshl,\poly[d, \keerr],\size{\X}]\)-state.
  \end{proof}

  We are going to show that there exists a key \(\key' = b \in \{0,1\}\) such that
  the probability of the key \(b\) in the key distribution of \(\jjointstate\)
  is larger than \(1 - \keerr\).
  By contradiction, assume that for both \(b=0\) and \(b=1\), we have that the
  probability of this key is smaller than \(1 - \keerr\).
  By considering the complementary events, we have that:
  \begin{align*}
    &\norm{\Pik[0] \jjointstate}^2 \geq \keerr, \textnormal{ and }\\
    &\norm{\Pik[1] \jjointstate}^2 \geq \keerr.
  \end{align*}
  Let \(\jjointstate[, \key'=b]\) be the residual state conditioned on the key
  equal to \(b\).
  Then, it follows that \(\jjointstate[, \key'=b]\) is a
  \(\goodstate[\Y, \epshl / \keerr, \poly[d, 1 / \keerr], \size{\X}]\)-state
  for both \(b=0\) and \(b=1\) because
  \begin{enumerate}
    \item \(\jjointstate[, \key'=b]\) is \(\poly[d, \frac{1}{\keerr}]\)-sparse since
      \(\jjointstate\) is \(\poly[d, \frac{1}{\keerr}]\)-sparse and
      \(\supp^{\widetilde{H}}(\jjointstate[, \key'=b]) \subseteq \supp^{\widetilde{H}}(\jjointstate)\).
    \item \(\jjointstate[, \key'=b]\) is \(\epshl/\keerr\)-light because:
    \begin{align*}
      \pr[\text{Not measuring \(\hat{0}\) in
      \(\jjointstate[, \key'=b]\)}] &= \pr\left[\text{Not measuring \(\hat{0}\) in
                                      \(\jjointstate[, \key'=b]\)} \given{\key' = b} \right] \\
                                    &= \frac{\pr\left[\text{Not
                                      measuring \(\hat{0}\) in \(\jjointstate\) and
                                      } \key' = b \right]}{\pr[\key'=b]} \\
                                    & \leq \frac{\pr\left[\text{Not measuring
                                      \(\hat{0}\) in
                                      \(\jjointstate\)}\right]}{\pr[\key'=b]} \\
                                    & \leq \epshl/\keerr,
    \end{align*}
      where the last inequality comes from the fact that \(\jjointstate\) is
      \(\epshl\)-light and
      \(\pr[\key'=b] = \norm{\Pik[b] \jjointstate}^2 \geq \keerr\).
  \end{enumerate}

  Then \Cref{conjecture:pcc} implies that the states \(\jjointstate[, \key'=0]\)
  and \(\jjointstate[, \key'=1]\) are compatible, which means that there exists
  two different states \(w^{0}, w^{1} \in W'_{\alice}EW\) and an oracle
  \(\hat{h}\) such that \(\hat{h}\) is consistent with \(w^{0}\) and \(w^{1}\).
  And for this specific oracle, \(w^{0}\) outputs the key \(0\) and \(w^{1}\)
  outputs the key \(1\), both with non-zero probability.
  However, Bob's key has already been computed by \Cref{lemma:bobmeasures}, and
  is fixed to some \(\kB \in \{0,1\}\).
  Thus, there is an oracle such that Bob outputs \(\kB\).
  Plus, for this specific oracle, Alice outputs key \(0\) with non-zero
  probability, and outputs key \(1\) with non-zero probability as well.
  Hence there is an execution of the protocol such that Bob outputs \(\kB\) and
  Alice outputs \(\kA = 1 - \kB\), which breaks the perfect completeness of the
  protocol.

  We now show that the key computed by Eve is
  the same as hypothetical Alice's key, defined by
  \(\khA = \Aop' \left(\ket{\realA}_{AH} \otimes \ket{\mi}_{M}\right)\), that is
  the key that Alice would have computed if the protocol had continued
  normally.
  Since the protocol is perfect, we have that \(\khA = \kB\).
  Recall that Eve's key is computed from the state
  \(\jjointstate = \Aop \left(\jointstate_{W'_{\alice}EW\widetilde{H}} \otimes \ket{\mi}_{M}\right)\),
  and the state \(\jointstate_{W'_{\alice}EW\widetilde{H}}\) is a superposition
  of all of Alice's internal states that are consistent with Eve's view so far.
  The state $\ket{\mi}_{M}$ corresponds to the real message that Bob sent to
  Alice.
  First note that Alice will never output \(\kA = \bot\), because the state of
  Eve consists of a superposition of Alice's states that are consistent with the
  transcript.
  Indeed, if the message \(\ket{\mi}_{M}\) from Bob is inconsistent with the
  oracle, Alice is not able to detect it as she does not query the oracle in
  \(\Aop\).
  Note that the real oracle used in the protocol is one of the oracles in the superposition of oracles that are consistent with Eve's view.
  Also, for a fixed oracle, the key is computed deterministically by the perfectness of the protocol, and thus Eve's key is equal to Bob's hypothetical key with probability \(\norm{\Pik\jjointstate}^{2}\) over the random oracles.
  This shows that Eve succeeds with probability at least \(1 - \keerr\).

\end{proof}

\subsubsection{Part 2: Making Alice Agrees on the Same Key as Bob}
\label{section:breakingccqc}

Using \Cref{lemma:simulation}, we can now prove \Cref{lemma:breakingccqc}.

\begin{proof}[Proof of \Cref{lemma:breakingccqc}]
  Let \((\alice,\bob)\) be a CC1QM-KA protocol where Alice and Bob make at most
  \(d\) queries to an oracle \(h:\mathcal{X} \rightarrow \mathcal{Y}\), and
  Alice does not query the oracle after receiving the last message.
  Consider the following construction for \(\eve\):
  \begin{construction}
    \label{construction:CC1QM-KA}
    Input: \(\keeps, \keerr\)
    \begin{enumerate}
      \item \(\eve\) applies the quantum \(\epshl\)-heavy query learner of
      \Cref{construction:simulation} to compute a state
      \(\ket{\simulatedA}_{W'_{A}}\ket{\simulatedB}_{W'_{B}}\ket{\h}_{H}\) which
      corresponds to a simulation of the internal state of Alice and Bob after
      the classical communication part of the protocol.
      \item Let \(\Aop\) be the operations that Alice applies at the end of the
      protocol after receiving the message \(\m\) from Bob.
      Then, \(\eve\) outputs the resulting key \(k_E\) of
      \(\Pik \Aop \kb{\simulatedA} \kb{\h} \m (\Aop)^{\dagger}\), where \(\m\)
      is the quantum message Bob sends to Alice.
      \item Writing
      \(\simulatedoutput_{EM} = \frac{\simulatedoutputu}{\norm{\simulatedoutputu}}\),
      where
      \(\simulatedoutputu = (\Aop)^{\dagger} \Pik \Aop \kb{\simulatedA} \kb{\h} \m\), \(\eve\) sends the resulting state \(\Tr_E\left(\simulatedoutput_{EM}\right)\) to Alice, where she traces out everything but the register that contains the message.
    \end{enumerate}
  \end{construction}
  In the last part of the construction, \(\eve\) applies the operator
  \(\left(\Aop\right)^{\dagger}\) to uncompute Alice's operation before sending
  her state to Alice.
  Note that in step 2, we use the fact that Alice and Bob's states are
  unentangled, as shown by \Cref{lemma:independency}

  Now, we prove that \Cref{construction:CC1QM-KA} succeeds with probability at least \(1 - \keerr\).
  Using \Cref{lemma:simulation}, we have that \(\eve\) finds the right key \(k\) in Step \(2\) with probability at least \(1 - \keerr\).
  Writing \(\psi = \sum_{i} \qi \ket{\mi}\), this means that
  \begin{equation}
    \label{eq:find}
    \forall i, \norm{\Pik \Aop \ket{\simulatedA} \ket{\h} \ket{\mi}}^2 \geq 1 - \lambda.
  \end{equation}

  We write \(\simulatedm = \Tr_E\left(\simulatedoutput_{EM}\right)\) the message
  that Eve sends to Alice.
  The first thing that we want to show is that the message \(\simulatedm\) from \(\eve\)
  is ``close'' to the real message \(\m\) from Bob.
  More precisely, we will show that:
  \begin{equation}
    \label{eq:simulatedm}
    \forall i, \bra{\mi}\simulatedm \ket{\mi} \geq 1 - \lambda.
  \end{equation}

  For every \(i\), we have that:
  \begin{align*}
    \bra{\mi}\simulatedm\ket{\mi} &=\bra{\mi} \Tr_E\left(\simulatedoutput_{EM}\right) \ket{\mi}\\
                                  &=\bra{\mi} \Tr_E\left(\frac{\simulatedoutputu}{\norm{\simulatedoutputu}}\right) \ket{\mi}\\
                                  &\geq \bra{\mi} \Tr_E\left(\simulatedoutputu  \right) \ket{\mi}\\
                                  &= \Tr \left( \kb{\mi} \Tr_E \left( \simulatedoutputu \right) \right)\\
                                  &= \Tr \left( I_E \otimes \bra{\mi}_M \simulatedoutputu  \cdot I_E \otimes \ket{\mi}_M \right)\\
                                  &\geq \Tr \left( \bra{\simulatedA} \bra{\h} \otimes \bra{\mi}_M \simulatedoutputu \cdot \ket{\simulatedA} \ket{\h} \otimes \ket{\mi}_M \right),
  \end{align*}
  where we used elementary properties of the trace operator.
  Next, we have that
  \[
    \Tr \left( \bra{\simulatedA} \bra{\h} \otimes \bra{\mi}_M \simulatedoutputu \cdot \ket{\simulatedA} \otimes \ket{\mi}_M \right) = \bra{\simulatedA} \otimes \bra{\mi}_M \simulatedoutputu \cdot \ket{\simulatedA} \ket{\h} \otimes \ket{\mi}_M,
  \]
  since the right term is a pure state.
  Replacing \(\simulatedoutputu\) with its value, we have:
  \begin{align*}
    \bra{\mi}\simulatedm\ket{\mi} &\geq\bra{\simulatedA} \bra{\h} \bra{\mi} \left((\Aop)^{\dagger} \Pik \Aop \kb{\simulatedA} \kb{\h} \m \right) \ket{\simulatedA} \ket{\h} \ket{\mi} \\
                                  &= \sum_{j} \qi[j] \bra{\simulatedA} \bra{\h} \bra{\mi} \left((\Aop)^{\dagger} \Pik \Aop \kb{\simulatedA} \kb{H} \kb{\mi[j]} \right) \ket{\simulatedA} \ket{\h} \ket{\mi} \\
                                  &= \bra{\simulatedA} \bra{\h} \bra{\mi} \left((\Aop)^{\dagger} \Pik \Aop \ket{\simulatedA} \ket{\h} \ket{\mi} \right) \\
  &=\norm{(\Aop)^{\dagger}\Pik \Aop \ket{\simulatedA}\ket{\h}\ket{\mi}}^2\\
  &= \norm{\Pik \Aop \ket{\simulatedA}\ket{\h}\ket{\mi}}^2 \\
  &\geq 1 - \lambda,
  \end{align*}
  where the last inequality comes from \Cref{eq:find}.

  Now fix \(i\).
  We write \(\ket{\realoutput} = \Aop \ket{\realA} \otimes \ket{\h} \otimes \ket{\mi}\) where \(\ket{\realA}\) corresponds to Alice's \emph{real} register before receiving the message \(\m\).
  Since \(\Pik\) is a projector and \(\Pik\ket{\realoutput} = \ket{\realoutput}\) from perfect correctness, we can write it:
  \[
  \Pik = \kb{\realoutput} + \sum_i \ketbra{\sigma_i}{\sigma_i},
  \]
  where the \(\sigma_i\) are such that \(\bk{\sigma_{i}}{\realoutput} = 0\).

  We write:
  \[
  \simulatedm = \alpha \kb{\mi} + \beta \rho,
  \]
  where \(\rho = \sum_j p_j \ketbra{\Psi_j}{\Psi_j}\) is a mixed state such  that \(\bra{\mi}\rho\ket{\mi} = 0\).

  For every \(\ket{\mi}\), we have that:
  \begin{align*}
    &\Tr\left(\Pik \Aop (\kb{\realA} \otimes \kb{\h} \otimes \simulatedm) \left(\Aop\right)^{\dagger}\right) \\
    &= \Tr\left(\left(\kb{\realoutput} + \sum_i \ket{\sigma_i}\bra{\sigma_i}\right) \left(\Aop \kb{\realA} \otimes \kb{\h} \otimes \left(\alpha \kb{\mi} + \beta \rho\right) \left(\Aop\right)^{\dagger}\right)\right) \\
    &= \Tr\left(\left(\kb{\realoutput} + \sum_i \ket{\sigma_i}\bra{\sigma_i}\right) \left(\alpha\kb{\realoutput} + \beta \Aop \kb{\realA} \otimes \kb{\h} \otimes \rho \left(\Aop\right)^{\dagger} \right)\right) \\
    &= \alpha \braket{\realoutput|\realoutput}^2 +
      \beta
      \bra{\realoutput}
      \Aop \kb{\phi_A} \otimes \kb{\h} \otimes \rho(\Aop)^\dagger \ket{\realoutput} \\
    & \quad+ \alpha \sum_i |\braket{\realoutput|\sigma_i}|^2 +
      \beta\sum_i \bra{\sigma_i} \Aop \kb{\realA} \otimes \kb{\h} \otimes \rho (\Aop)^\dagger \ket{\sigma_i}
    \\
    &= \alpha +
      \beta\sum_i \bra{\sigma_i} \Aop \kb{\realA} \otimes \kb{\h} \otimes \rho (\Aop)^\dagger \ket{\sigma_i} \\
    &\geq \alpha \\
    &\geq 1 - \lambda,
  \end{align*}
  where the fourth equality comes from the fact that
  \begin{align*}
    \bra{\realoutput}
    \Aop \kb{\realA}\otimes \rho
    (\Aop)^\dagger\ket{\realoutput}
    =
    \bra{\realA}\bra{\mi}(\Aop)^\dagger
    \Aop \kb{\realA}\otimes \rho (\Aop)^\dagger \Aop \ket{\realA}\ket{\mi}  =
    \bra{\mi}\rho\ket{\mi} = 0,
  \end{align*}
  and that \(\bk{\sigma_{i}}{\realoutput} = 0\).
  The first inequality comes from the fact that the terms in the sum are positive,
  because they correspond to the probability of measuring the state
  \(\kb{\realA} \otimes \kb{\h} \otimes \rho \) using the projection
  \(\Aop \braket{\sigma_{i}}{\Aop}^\dagger\), and the last inequality comes from
  \Cref{eq:simulatedm}.

  This means that Alice measures the key \(k\) with probability at least \(1-\lambda\)
  when receiving the message \(\simulatedm\) from \(\eve\) and for pure message
  \(\ket{\mi}\), and if this is the case the meet-in-the-middle attack is a success.
  Since this is true for all of the \(\ket{\mi}\), it also follows for \(\m\) by
  convexity.
  This concludes the proof.

\end{proof}

\subsection{Impossibility of quantum public-key encryption with classical keys}
\label{sec:pke}

In this section, we show that the (conditional) impossibility of CC1QM-KA protocols
proven above also implies a (conditional) impossibility for quantum public key
encryption (qPKE) with classical public keys, but ciphertexts can be quantum states.

We first define the notion of qPKE with classical keys which is modified from
the notion of qPKE with quantum keys given in~\cite{BGHMSVW23}, then prove our
impossibility.

\begin{definition}[Public-key encryption with classical public keys]
\label{def:qpke}
Public-key encryption with classical public keys (qPKE) consists of three
algorithms with the following syntax:
\begin{itemize}[label=\(\bullet\)]
  \item \((\pk, \sk) \gets \qgen(1^{\secpar}):\) a quantum algorithm, which takes as
    input the security parameter and outputs a classical key pair
    \((\pk, \sk)\).
  \item \(\qc \gets \qenc(\pk,m)\): a quantum algorithm, which takes as input a
    classical public key \(\pk\), a plaintext \(m\), and outputs a possibly
    quantum ciphertext \(\qc\).
  \item \(m / \bot \gets \qdec(\sk, \qc)\): a quantum algorithm, which takes as input a
    decryption key \(\sk\), a ciphertext \(\qc\), and outputs a classical
    plaintext \(m\) or a distinguished symbol \(\bot\) indicating decryption
    failure.
  \end{itemize}
  Furthermore, in the quantum random oracle model, we allow these algorithms to
  make quantum queries to a random function \(H\).
  We also allow these algorithms to be inefficient, but can only make at most
  polynomially (in the security parameter) many queries to the random oracle.
\end{definition}
We say that a qPKE scheme is \emph{perfectly correct} if for every message
\(m \in \bin^*\) and any security parameter \(\secpar \in \NN\), the following
holds:
\[
\Pr\left[
\qdec(\sk, \qc) = m
\ \middle\vert
\begin{array}{c}
  (\pk, \sk) \gets \qgen(1^{\secpar}) \\
  \qc \gets \qenc(\pk, m)
\end{array}
\right] = 1,
\]
where the probability is taken over the randomness of \(\qgen\), \(\qenc\) and
\(\qdec\).

We next define indistinguishability security of one-bit qPKE
in~\cref{def:qpke-ow}.
When considering one-bit encryption this notion coincides with the one-way
security notion, which is considered the weakest security notion of
encryption.
Thus, using this notion makes our negative result stronger.

\begin{definition}
  \label{def:qpke-ow}
  A one-bit qPKE scheme with classical public keys is IND-CPA secure if for
  every \(\qpt\) adversary \(\adv\), for any \(\secpar \in \NN\), there exists a
  negligible function \(\negl\) such that
  \begin{equation*}
    \prob{\mathtt{IND-CPA}(\secpar, \adv) = 1} \leq \frac{1}{2} + \negl,
  \end{equation*}
  where \(\mathtt{IND-CPA}(\secpar, \adv)\) is the following experiment:
  \begin{enumerate}
    \item The challenger chooses a random key pair
      \((\pk, \sk) \gets \qgen(1^{\secpar})\), and sends \(\pk\) to the adversary
      \(\adv\).
    \item \(\adv\), upon receiving the public key \(\pk\), sends
      two bits \(m_{0}, m_{1} \in \bin\) to the challenger.
    \item The challenger samples a random bit \(b \sample \bin\), and sends
    \(\qc \gets \qenc(\pk, m_{b})\) to \(\adv\).
    \item \(\adv\) responds with a guess \(b'\) for \(b\).
    \item The challenger outputs \(1\) if \(b' = b\), and \(0\) otherwise.
  \end{enumerate}
\end{definition}

Since the existence of an IND-CPA secure qPKE scheme with classical public keys
in the QROM implies the existence of a CC1QM-KA protocol in the QROM, we also
obtain the following result.
\begin{corollary}
  Assuming \Cref{conjecture:pcc} is true, there is no IND-CPA secure qPKE scheme with
  classical public keys in the QROM, where the decryption algorithm does not
  query the random oracle.
\end{corollary}
\begin{proof}
  By contradiction, let \(\Pi = (\qgen, \qenc, \qdec)\) be a qPKE scheme with
  classical public keys and assume it is IND-CPA secure.
  We construct a two-message one-bit CC1QM-KA protocol \(\tilde{\Pi}\), where
  the first message from Alice to Bob is classical and the second message from
  Bob to Alice is quantum, as follows.
  \begin{enumerate}
    \item Alice generates \((\pk, \sk) \leftarrow \kgen(1^{\secpar})\), and sends \(\pk\)
      to Bob.
    \item Bob generates uniformly at random a secret key \(\key \in \bin\) and
      computes \(\qc \gets \qenc(\pk, \key)\), and sends \(\qc\) to Alice.
    \item Alice recovers the common key by computing \(\key \gets \qdec(\sk, \qc)\).
  \end{enumerate}
  It is easy to see that \(\tilde{\Pi}\) is a secure CC1QC-KM protocol in the
  QROM if \(\Pi\) is IND-CPA secure.
  Furthermore, if \(\Pi\) is perfectly correct, \(\tilde{\Pi}\) is also
  perfectly correct.
  Finally, if \(\qdec(\cdot,\cdot)\) does not query the oracle, then Alice in
  the last step of \(\tilde{\Pi}\) does not query the oracle as well.
  This contradicts~\cref{theorem:mainka} and concludes our proof.
\end{proof}

\begin{remark}
  We note that our impossibility of CC1QM-KA is the strongest possible
  impossibility (conditioned on the assumption that~\cref{conjecture:pcc} is
  true), in the sense that the adversary can find the shared key and
  maintain the correctness of the protocol (that is, Alice and Bob can still
  find the shared key), while the usual security definition only asks the
  adversary to find the key of one of two parties.
  This strong impossibility allows us to rule out the possibility of qPKE in the
  QROM with stronger requirements, for example, qPKE with decryption error
  detectability as defined in~\cite{KMNY23}.
\end{remark}

\ifdraft
\else
\section*{Acknowledgments}
This work is part of HQI initiative (www.hqi.fr) and is supported by France 2030 under the French National Research Agency award number ANR-22-PNCQ-0002.

ABG is supported by ANR JCJC TCS-NISQ ANR-22-CE47-0004  and QHV was also partially supported by the same grant.  

\fi
\bibliographystyle{alpha}
\renewcommand{\doi}[1]{\url{#1}}
\bibliography{./cryptobib/abbrev3,./cryptobib/crypto,references}

\end{document}